\documentclass{article}

\usepackage{arxiv}

\usepackage[utf8]{inputenc} 
\usepackage[T1]{fontenc}    
\usepackage{hyperref}       
\usepackage{url}            
\usepackage{booktabs}       
\usepackage{amsfonts}       
\usepackage{nicefrac}       
\usepackage{microtype}      
\usepackage{lipsum}

\usepackage{cite}
\usepackage{amsmath,amssymb,amsfonts}
\usepackage{graphicx}
\usepackage{algorithm}
\usepackage[noend]{algpseudocode}
\usepackage{hyperref}
\usepackage{textcomp}
\usepackage{threeparttablex} 

\usepackage{mathtools}
\usepackage{caption}
\usepackage{float}
\usepackage{stmaryrd}
\usepackage{epstopdf}
\usepackage{multirow}
\usepackage{pifont}
\usepackage{caption}
\usepackage{subcaption}

\newcommand{\R}{{\mathbb{R}}}

\newcommand{\Cs}{{\mathcal{C}}}
\newcommand{\N}{{\mathbb{N}}}

\newcommand{\X}{{\mathbf{X}}}

\newcommand{\U}{{\mathbf{U}}}

\newcommand{\F}{{\lozenge}}
\newcommand{\G}{{\Box}}
\newcommand{\cen}{\mathsf c}
\newcommand{\rad}{\mathsf r}
\newcommand{\W}{{\mathcal{W}}}
\newcommand{\B}{{\mathcal{B}}}
\newcommand{\true}{{\mathsf{true}}}
\newcommand\mydots{\hbox to 1em{.\hss.\hss.}}

\newtheorem{theorem}{Theorem}[section]

\newtheorem{assumption}{Assumption}

\newtheorem{definition}[theorem]{Definition}
\newtheorem{lemma}[theorem]{Lemma}
\newtheorem{remark}[theorem]{Remark}
\newtheorem{problem}[theorem]{Problem}
\newtheorem{proof}[theorem]{Proof}

\title{Control Barrier Functions for the Full Class of Signal Temporal Logic Tasks using Spatiotemporal Tubes
\thanks{ This work was supported in part by the SERB Start-Up Research Grant; in part by the ARTPARK. The work of Ratnangshu Das was supported by the Prime Minister’s Research Fellowship from the Ministry of Education, Government of India.}
}

\author{
 Ratnangshu Das \\
  Robert Bosch Centre for Cyber-Physical Systems\\
  IISc, Bengaluru, India\\
  \texttt{ratnangshud@iisc.ac.in} \\
   \And
 Subhodeep Choudhury \\
  Robert Bosch Centre for Cyber-Physical Systems\\
  IISc, Bengaluru, India\\
  \texttt{subho02.dc@gmail.com} \\
  \And
 Pushpak Jagtap \\
  Robert Bosch Centre for Cyber-Physical Systems\\
  IISc, Bengaluru, India\\
  \texttt{pushpak@iisc.ac.in} \\
}

\begin{document}
\maketitle

\begin{abstract}
This paper introduces a new framework for synthesizing time-varying control barrier functions (TV-CBFs) for general Signal Temporal Logic (STL) specifications using spatiotemporal tubes (STT). We first formulate the STT synthesis as a robust optimization problem (ROP) and solve it through a scenario optimization problem (SOP), providing formal guarantees that the resulting tubes capture the given STL specifications. These STTs are then used to construct TV-CBFs, ensuring that under any control law rendering them invariant, the system satisfies the STL tasks. We demonstrate the framework through case studies on a differential-drive mobile robot and a quadrotor, and provide a comparative analysis showing improved efficiency over existing approaches.
\end{abstract}

\section{Introduction}
Ensuring that autonomous systems, such as robots and unmanned vehicles, operate safely and reliably while performing complex missions is a central challenge in modern control theory. Signal Temporal Logic (STL) \cite{STL_OdedMaler} has emerged as a powerful and expressive formal language for defining complex, time-dependent tasks that go well beyond classical goals like stabilization and trajectory tracking. 
As a result, STL has gained increasing attention across various domains, including neural and deep learning \cite{STL_RNN, STLnet}, reinforcement learning \cite{STL_RL}, learning from demonstration \cite{STL_lfd}, and planning and control \cite{2025arXiv250713225S, STL_multi_chuchu}. 
However, despite its expressive power, synthesizing controllers that guarantee satisfaction of STL specifications remains both computationally demanding and algorithmically challenging.

Existing approaches for STL-constrained control, like Mixed-Integer Programming (MIP) \cite{STL_reactive_MILP} and gradient-based methods \cite{STL_grad}, often rely on optimization and scale poorly with the complexity of the STL formula and system dynamics. 
Model Predictive Control (MPC) has also been extended to incorporate STL constraints within a receding-horizon framework \cite{STL_MPC, STL_MPC3}, yet these methods often become computationally intractable for long or complex tasks. 
Symbolic control provides formal correctness guarantees but typically depends on state-space abstractions \cite{tabuada2009verification}, making it highly susceptible to the curse of dimensionality.
An alternative line of work uses Prescribed Performance Control (PPC) to enforce performance "funnels" on STL robustness semantics \cite{STL_funnel,liu2022compositional}. Although computationally efficient, these methods apply only to a fragment of STL specifications. More recently, a spatiotemporal tube (STT) based approach was proposed in \cite{das2025approximation}, providing an approximation-free feedback control law for general STL tasks. However, it is limited to fully actuated systems and relies on hyper-rectangular STTs that are computationally expensive and incompatible with smooth time-varying control barrier functions (TV-CBFs).

Recently, Control Barrier Functions (CBFs) \cite{CBF} have attracted considerable attention as an effective means of enforcing safety constraints with formal guarantees. 
They provide a computationally efficient framework for controller synthesis while ensuring forward invariance of a safe region in the state space. 
A natural and promising research direction has therefore been the integration of these two formalisms. \cite{CBF_STL_Belta} use CBFs to enforce STL specifications in continuous time by treating STL invariance requirements as safety constraints.
Similarly, \cite{CBF_STL} demonstrated how time-varying CBFs (TV-CBFs) can be constructed to under-approximate the predicate functions of the STL formula at selected time instances, and thus be used to satisfy STL tasks. 
These formulations provide an elegant, abstraction-free link between high-level logical specifications and low-level control actions. However, they remains limited to a fragment of STL and struggles to handle more complex formulas with nested temporal operators.

TV-CBFs \cite{TVCBF} are well suited for enforcing time-dependent safety constraints, and their connection to STL semantics, established in \cite{CBF_STL}, provides a promising foundation. Building on this insight, we aim to develop a framework that systematically handles the full class of STL tasks. The key limitation, therefore, lies not in the use of CBFs themselves, but in the lack of a synthesis methodology for general STL specifications. 
Recent works have explored learning-based strategies, framing CBF synthesis as a machine learning problem that leverages safe and unsafe data samples \cite{learn-CBF-1} or expert demonstrations \cite{learn-CBF-ED}. However, these data-driven methods often lack formal correctness guarantees and can be computationally expensive. 
Recent learning-based approaches attempt to learn CBFs from safe–unsafe samples \cite{learn-CBF-1} or expert demonstrations \cite{learn-CBF-ED}, but they lack formal guarantees, are computationally demanding, and do not address STL tasks.

In this paper, we propose a novel framework that addresses this synthesis challenge by leveraging the STT framework \cite{STT, das2025spatiotemporal}. 
The main contributions of this paper are threefold:
(i) given an STL specification, we formulate the synthesis of the corresponding STT as a robust optimization problem (ROP) by leveraging the STL robustness metric \cite{STL_robust};
(ii) we reformulate the ROP as a computationally feasible scenario optimization problem (SOP) by sampling points in time and state space, and provide formal guarantees that the resulting STT correctly captures the STL specification; and
(iii) we use the synthesized STT to construct a TV-CBF and show that, under a control law rendering this TV-CBF invariant, the system satisfies the full class of STL specifications.
A key technical innovation of our framework is the use of STT with spherical cross-sections, which, in contrast to our earlier hyper-rectangular formulation, significantly reduces computational complexity and allows smoother CBF construction. The effectiveness and scalability of the proposed method are demonstrated through case studies involving a differential-drive mobile robot and a quadrotor operating in a 3D environment. We also compare our approach against MILP-, MPC-, CBF-, and PPC-based methods, and the results highlight that our approach achieves faster computation, while handling more complex STL specifications.

\section{Preliminaries and Problem Formulation} \label{sec:prob}
\textit{Notations}: For $a,b\in\N$ with $a\leq b$, we denote the closed interval in $\N$ as $[a;b] := \{a, a+1, \ldots, b\}$. We define the ball $\B(\cen, \rad) = \{ x \in \R^n \mid \|x - \cen\| < \rad \}$, characterized by center $\cen \in \mathbb{R}^n$ and radius $\rad \in \mathbb{R}^+$. All other notation in this paper follows standard mathematical conventions.

\subsection{System Definition}
We consider a continuous-time, nonlinear, time-varying, control-affine system $\mathcal{S}$ defined by:
\begin{align}
    \mathcal{S}: \dot{x} = f(t,x) + g(t,x)u,
    \label{eqn:sysdyn}
\end{align}
where $ x(t) \in \X \subseteq \R^n $ is the system state, $u(t) \in \U \subseteq \R^m $ is the control input, $\X$ is a compact state space. $ f: \R_0^+ \times \X \rightarrow \R^n $ and $ g:\R_0^+ \times \X \rightarrow \R^{n \times m} $ are locally Lipschitz continuous.

\subsection{Signal Temporal Logic (STL)}
An STL \cite{STL_OdedMaler} formula is built from predicates $\mathsf{p}$, Boolean connectives, and temporal operators. Given a predicate function $h:\R^n \rightarrow \R$,  
$\mathsf{p} := {\true} \text{, if $h(s) \geq 0$, and} \ \mathsf{false} \text{ if $h(s) < 0$}.$
An STL formula $\phi$ is defined recursively as:
\begin{align}\label{eqn:stl}
    \phi := {\true} \ | \ \mathsf{p} \ | \ \neg \phi \ | \ \phi_1 \land \phi_2 \ | \  \phi_1\mathcal{U}_{[a, b]}\phi_2,
\end{align}
where $\mathsf{p}$ is a predicate, and $\phi_1,\phi_2$ are STL formulae. Here, $\neg$ and $\land$ denote logical negation and conjunction, and $\mathcal{U}_{[a,b]}$ is the \textit{until} operator over the time interval $[a,b]$ with $0 \leq a \leq b < \infty$. The \textit{disjunction}, \textit{eventually}, and \textit{always} operators can be respectively derived as $\phi_1 \lor \phi_2 := \neg (\neg \phi_1 \land \neg \phi_2)$, $\F_{[a,b]}\phi := \true \ \mathcal{U}_{[a,b]} \phi$, and $\G_{[a,b]}\phi := \neg \F_{[a,b]} \neg \phi$.

The satisfaction relation $(x,t)\models\phi$ denotes if a signal $x: \R_{\geq 0} \rightarrow \R^n$, possibly a solution of \eqref{eqn:sysdyn}, satisfies an STL formula $\phi$ starting from time $t$. The STL semantics \cite{STL_OdedMaler} for a signal $x$ is recursively given by:
\begin{align*}
    &(x,t)\models\mathsf{p} &&\Leftrightarrow \  h(x(t))\ge0, \\
    &(x,t)\models\neg\phi &&\Leftrightarrow \  \neg((x,t)\models\phi), \\
    & (x,t)\models\phi_1\land\phi_2 &&\Leftrightarrow \  (x,t)\models\phi_1\land (x,t) \models\phi_2, \\
    &(x,t)\models \phi_1{\mathcal{U}}_{[a, b]}\phi_2 &&\Leftrightarrow \  \exists t_1 \in [t+a,t+b], (x,t_1)\models\phi_1 \\
    & && \quad \ \ \land \forall t_2 \in [t+a,t_1], (x,t_2) \models \phi_2.
\end{align*}
A signal $x$ satisfies $\phi$, denoted by $x \models \phi$, iff $(x,0) \models \phi$. 

The robustness metric $\rho^{\phi}(x,t)$ \cite{STL_robust} quantifies the degree of satisfaction or violation of $\phi$. $\rho^{\phi}(x,t)\geq 0$ implies $(x,t)\models \phi$, and larger values correspond to stronger satisfaction. It is defined recursively as:
\begin{subequations}\label{eqn:stl_rho}
\begin{align}
    &\rho^{\mathsf{p}}(x,t)= h (x(t)), \label{eqn:stta}\\
    &\rho^{\neg\phi} (x,t)= -\rho^{\phi} (x,t), \\
    &\rho^{\phi_1\land\phi2} (x,t)= \min \big(\rho^{\phi_1} (x,t), \rho^{\phi_2} (x,t)\big),\\
    &\rho^{\phi_1\lor\phi2} (x,t)= \max \big(\rho^{\phi_1} (x,t), \rho^{\phi_2} (x,t)\big),\\
    &\rho^{{\F}_{[a, b]}\phi} (x,t)= \max_{t_1 \in [t+a,t+b]} \rho^{\phi} (x,t_1),\\
    &\rho^{{\G}_{[a, b]}\phi} (x,t)= \min_{t_1 \in [t+a,t+b]} \rho^{\phi} (x,t_1),\\
    &\rho^{\phi_1{\mathcal{U}}_{[a, b]}\phi_2} (x,t) = \max _{t_1 \in [t+a, t+b]} \min \big( \rho^{\phi_1}(x,t_1), \notag \\
    & \qquad \qquad \qquad \qquad \min_{t_2 \in [t+a,t_1]} \rho^{\phi_2} (x,t_2)\big).
    \label{eqn:stte}
\end{align}
\end{subequations}
For simplicity, we denote the robustness at $t = 0$ by $\rho^{\phi}(x)$. We consider a finite time horizon $[0,t_f]$ over which the specification $\phi$ is to be realized, i.e., the signal $x:[0,t_f] \rightarrow \R^n$ is such that $x \models \phi$.

\subsection{Time-varying Control Barrier Functions}
\label{subsec:desc}
Time-varying control barrier functions (TV-CBFs) offer a natural way to enforce time-dependent constraints, such as those specified by Signal Temporal Logic (STL) specifications. Following the framework in \cite{TVCBF}, a TV-CBF is a differentiable function $b: {\left[0, t_f\right]} \times \X \rightarrow \R$, where $\X \subset \R^n$ is the state space. The corresponding time-varying safe set $\Cs(t)$ is given by the 0-superlevel set of $b(t,x)$ as
\begin{subequations}\label{eq:safeset}
\begin{align}
    \Cs(t) &= \{x \in \X | b(t,x) \geq 0\} \\
    \partial\Cs(t) &= \{x \in \X | b(t,x) = 0\} \\
    \text{Int}(\Cs(t)) &= \{x \in \X | b(t,x) > 0\},
\end{align}
\end{subequations}
where $\partial \Cs$ and Int$(\Cs)$ are the boundary and interior of $\Cs$.

Now, $b(t,x)$ is considered a candidate TV-CBF if $\Cs(t)$ is non-empty for the entire time horizon. A trajectory $x:[0, t_f] \rightarrow \R^n$ with $ x(t) \in \Cs(t) $ for all $ t \in \left[0, t_f\right] $ ensures this.
\begin{definition}[Candidate Control Barrier Function \cite{CBF_STL}]
A differentiable function $ b : \left[0, t_f\right] \times \X \rightarrow \R $ is called a candidate control barrier function if, for any $x_0 \in \Cs(0)$, there exists an absolutely continuous trajectory $x:[0, t_f] \rightarrow \R^n $ with $ x(0) := x_0 $ such that $x(t) \in \Cs(t) $ for all $ t \in \left[0, t_f\right] $.
\end{definition}

We now define the notion of a valid control barrier function.
\begin{definition}[Valid Control Barrier Function]
    Consider the control-affine system $\mathcal{S}$ in \eqref{eqn:sysdyn} and the time-varying safe set $\Cs(t) \subseteq \X$ defined in \eqref{eq:safeset}. The function $b: [0, t_f] \times \X \rightarrow \R$ is a valid control barrier function if there exists a locally Lipschitz class-$\mathcal{K}$ function 
    \footnote{A class $\mathcal{K}$ function is a continuous, strictly increasing function $\alpha: [0, a) \to [0, \infty)$ with $\alpha(0) = 0$.}
    $\alpha: \R_0^+ \rightarrow \R$ such that, for all $(t,x) \in [0, t_f] \times \Cs(t)$,
    \begin{equation*}
        \sup_{u} \Big( \tfrac{\partial b(t,x)}{\partial x}^\top (f(t,x) + g(t,x)u) + \tfrac{\partial b(t,x)}{\partial t} \Big) \geq -\alpha(b(t,x)).
    \end{equation*}   
\end{definition}

The set of all safety-preserving control inputs at any time $t$ and in any state $x$ is defined as 
$\U_a(t,x) := \big\{ u \in \U \mid \tfrac{\partial b(t,x)}{\partial x}^\top (f(t,x) + g(t,x)u) + \tfrac{\partial b(t,x)}{\partial t} \geq -\alpha(b(t,x)) \big\}.$
A candidate CBF becomes a valid CBF if it guarantees the forward invariance of $\Cs(t)$ under the chosen control law.

\begin{definition}[Forward Invariance \cite{CBF_STL}]
A set $\Cs(t)$ $\subseteq \X$ is said to be forward invariant for the system \eqref{eqn:sysdyn}, under a control law $u(t,x)$, if for every initial condition $x(0) \in \Cs(0)$, there exists a unique trajectory $x:[0,t_f] \rightarrow \R^n$ with $x(0) = x_0$ satisfying $x(t) \in \Cs(t)$ for all $ t \in \left[0, t_f\right] $.
\end{definition}

The fundamental result connecting TV-CBFs to system safety is summarized in the following theorem. 
\begin{theorem}[\cite{glotfelter2017nonsmooth}]\label{thm:tvcbf}
    If $u(t,x) \in \U_a(t,x)$ is locally Lipschitz continuous in $x$ and piecewise continuous in $t$ and the (consequently) unique solutions to \eqref{eqn:sysdyn} are defined over the time interval $[0, t_f]$. Then the set $\Cs(t)$ is forward invariant for the control law $u(t,x)$ if $b(t, x)$ is a valid control barrier function.
\end{theorem}

\subsection{Problem Formulation}

However, given an STL task, designing a time-varying CBF $b(t,x)$ is particularly challenging. We now formally define the problem considered in this work.

\begin{problem}\label{prob1}
   Given the system $\mathcal{S}$ in \eqref{eqn:sysdyn} and an STL specification $\phi$ over time-interval $[0,t_f]$, defined in \eqref{eqn:stl}, the goal is to design a time-varying control barrier function, $b: \left[0, t_f\right] \times \X \rightarrow \R$, and derive a control $u(t,x)$ such that the controlled system trajectory $x$ satisfies the specification $x \models \phi$. 
\end{problem}

\section{Spatiotemporal Tubes}
To satisfy the given STL specification, we leverage the Spatiotemporal Tube (STT) framework. STTs act as dynamic structures in the state space that can effectively capture the STL task \cite{das2025approximation}. 

\begin{definition}[STT for STL Specification]\label{def:stt}
Consider an STL task $\phi$ over the time interval $[0, t_f]$. A time-varying region 
$$\Gamma(t) = \B\big(\cen(t),\rad(t)\big)$$
characterized by continuously differentiable centre $\cen:\R_0^+\rightarrow\R^n$ and radius $\rad:\R_0^+\rightarrow\R^+$, is a valid STT for the STL specification, if the following hold:
\begin{subequations}
\begin{align}\label{eqn:stt_stl}
    &\rad(\tau) > 0, \forall \tau \in [0, t_f], \\
    &\rho^\phi (x) > 0, \forall x: [0,t_f] \rightarrow \R^n, \text{ s.t. } x(\tau) \in \Gamma(\tau), \forall \tau \in [0, t_f].
\end{align}
\end{subequations}
Thus, the radius remains strictly positive and the robustness of the STL formula is positive for all trajectories inside the tube.
\end{definition}
Now, if the system is constrained within the STT:
\begin{align} \label{eqn:stt_constrain}
    x(\tau) \in \Gamma (\tau), \forall \tau \in [0, t_f],
\end{align}
then the STL specification $\phi$ is guaranteed to be satisfied.

\subsection{Construction of Spatiotemporal Tubes}\label{sec:sampling_based}
The goal in this section is to construct STT that guarantee satisfaction of the STL specification. We begin by fixing the functional form of the STT 
$$\Gamma(q_\cen, q_\rad, t) = \B(\cen(q_\cen, t), \rad(q_\rad, t))$$
with the center $\cen(q_\cen, t) = [\cen_1(q_{\cen,1}, t), \ldots, \cen_n(q_{\cen,n}, t)]^\top \in \R^n$ and the radius $\rad(q_\rad, t) \in \R$ defined as
\begin{align*}
    \cen_{i}(q_{\cen,i},t) &= \sum_{k=0}^{z_{i}} q_{\cen,i}^{(k)} p_{\cen,i}^{(k)}(t), \ i\in [1;n], \\
    \rad(q_\rad,t) &= \sum_{k=0}^{z_{\rad}} q_{\rad}^{(k)} p_{\rad}^{(k)}(t).
\end{align*}
Here, $z_i$ and $z_\rad$ are the number of coefficients, $p_{\cen,i}(t)=[p_{\cen,i}^{(1)}, ..., p_{\cen,i}^{(z_i)}]^\top$, $p_\rad(t)=[p_{\rad}^{(1)}, ..., p_{\rad}^{(z_\rad)}]^\top$ are user-defined continuously differentiable basis functions and $q_\cen = [q_{\cen,1}, \ldots, q_{\cen,n}]^\top \in \mathbb{R}^{nz_i}$ with $q_{\cen,i} = [q_{\cen,i}^{(1)}, ..., q_{\cen,i}^{(z_i)}]^\top \in \mathbb{R}^{z_i}$ and $q_\rad = [q_\rad^{(1)}, ..., q_\rad^{(z_\rad)}] \in \mathbb{R}^{z_\rad}$ denote the unknown coefficients.

$p_{\cen,i}(t)$ and $p_\rad(t)$ are user-defined continuously differentiable basis functions, and $q_\cen = [q_{\cen,1}, \ldots, q_{\cen,n}]^\top \in \R^n$, $q_\rad \in \R$ denote unknown coefficients.

To enforce the conditions in Definition~\ref{def:stt}, we formulate the following Robust optimization problem (ROP):
\begin{subequations} \label{eq:ROP}
\begin{align}
& \min_{[q_\cen, q_\rad, \eta]} \quad \eta, \quad  \textrm{s.t.}  \notag \\
& \forall \tau \in [0, t_f]: \ -\rad(q_\rad, \tau) + \rad_d \leq \eta, \label{subeqn:ROP1} \\
&  \forall x: [0,t_f] \rightarrow \R^n, \text{s. t., } x(\tau) = \cen(q_\cen, \tau) + \lambda\rad(q_\rad, \tau)s(\theta), \notag \\ 
& \hspace{2.5 cm} \forall (\theta,\lambda,\tau) \in [0,2\pi]^n \times [0,1] \times [0, t_f], \notag \\
& -\rho^{\phi} \big(x\big) \leq \eta, 
\label{subeqn:ROP3}
\end{align}
\end{subequations}
where 
\begin{align*}
    s(\theta) = &[\cos\theta_1, \sin\theta_1\cos\theta_2, \sin\theta_1\sin\theta_2\cos\theta_3, \ldots, \\
    &\quad \sin\theta_1\sin\theta_2\ldots\sin\theta_{n-1}\cos\theta_{n-1}, \\
    &\quad \sin\theta_1\sin\theta_2\ldots\sin\theta_{n-1}\sin\theta_{n-1}]^\top,
\end{align*}
parameterizes the unit sphere, and $\rad_{d} \in \R^+$ is a user-defined lower bound on the tube radius. 
\begin{remark}\label{rem:stt_constraints}
The condition \eqref{subeqn:ROP1} guarantees that the tube radius never shrinks to zero, ensuring a feasible safe set of non-zero volume at all times.  
The constraint in \eqref{subeqn:ROP3} enforces that the STL robustness remains positive inside the tube. Here, $s(\theta)$ covers all directions on the unit sphere, while $\lambda \in [0,1]$ scales radially outward from the center to the boundary. Together, they span the entire ball $\B(\cen(\tau), \rad(\tau))$, ensuring that any trajectory confined within the STT automatically satisfies the STL specification.
\end{remark}

It is straightforward to see that if the optimal solution $\eta^*$ to the ROP satisfies $\eta^* \leq 0$, then the conditions in Definition~\ref{def:stt} are guaranteed to hold.

The main difficulty, however, lies in the fact that the ROP in \eqref{eq:ROP} contains infinitely many constraints defined over continuous time and space. To make the problem tractable, we adopt a sampling-based approach for constructing the tube.

We introduce the augmented set $\W = [0,2\pi]^n \times [0,1] \times [0, t_f]$ and collect $N$ samples $w_s = (\theta_s, \lambda_s, \tau_s)$, $s \in [1;N]$. For each $w_s$, we define a ball $\B(w_s,\epsilon)$ of radius $\epsilon$, such that for all the points $(\theta, \lambda, \tau) \in \W$, there exists a $w_s$ that satisfies 
\begin{align}\label{eq:ball}
    \| (\theta, \lambda, \tau) - w_s \| \leq \epsilon. 
\end{align}  
The union of all such balls then forms a cover of the augmented set, i.e., $\bigcup_{s=1}^{N} \B(w_s,\epsilon) \supset \W$.

Based on this construction, we reformulate the ROP as the following scenario optimization problem (SOP):
\begin{subequations} \label{eq:SOP}
\begin{align}
& \min_{[q_\cen, q_\rad, \eta]} \quad \eta, \quad  \textrm{s.t.}  \notag \\
& \forall \tau_s \in [0, t_f]: \ -\rad(q_\rad, \tau_s) + \rad_d \leq \eta \label{subeqn:SOP1} \\
&  \forall x: [0,t_f] \rightarrow \R^n, \text{s. t., } x(\tau_s) = \cen(q_\cen, \tau_s) + \lambda_s\rad(q_\rad, \tau_s)s(\theta_s), \notag \\ 
& \hspace{2.5 cm} \forall (\theta_s,\lambda_s,\tau_s) \in [0,2\pi]^n \times [0,1] \times [0, t_f], \notag \\
& -\rho^{\phi} \big(x\big) \leq \eta.
\label{subeqn:SOP3}
\end{align}
\end{subequations}

The SOP \eqref{eq:SOP} now involves only finitely many constraints, while retaining the same structure as the original ROP \eqref{eq:ROP}. To ensure that the STT obtained from the SOP remain valid solutions of the ROP over continuous time and space, we impose the following assumption:
\begin{assumption} \label{assum:funlip}
    The functions $\cen(q_{\cen},t)$ and $\rad(q_{\rad},t)$ are Lipschitz continuous in $t$, with constants $\mathcal{L}_{\cen}$ and $\mathcal{L}_{\rad}$, respectively.
\end{assumption}

From \cite{STL_Lip}, we know that the STL robustness function $\rho^\phi(x)$ is Lipschitz continuous with respect to the signal, i.e., there exists $\mathcal{L}_\rho \in \R^+$ such that $|\rho^\phi(x)-\rho^\phi(y)| \leq \mathcal{L}_\rho \|x-y\|_\infty$ for all $x,y:[0,t_f] \rightarrow \R^n$. 
\begin{lemma}\label{thm:lip_rob}
    $\mu(\theta, \lambda)
    := -\rho^{\phi} ( x )$ with $x(\tau) = \cen(q_\cen, \tau) + \lambda \rad(q_\rad, \tau)s(\theta)$ for all $\tau \in [0,t_f]$, is Lipschitz continuous with the Lipschitz constant $\mathcal{L}_\mu \in \R^+$. 
    That is, for all $ (\theta, \lambda),(\hat{\theta}, \hat{\lambda}) \in [0,2\pi]^n \times [0,1]$, 
    $$\|\mu(\theta, \lambda)-\mu(\hat\theta, \hat\lambda)\| \leq \mathcal{L}_\mu \| (\theta, \lambda) - (\hat\theta, \hat\lambda) \|.$$
\end{lemma}
\begin{proof}
    See Appendix.
\end{proof}

In Theorem \ref{th:constr}, we show that even though the SOP \eqref{eq:SOP} is solved using a finite set of sampled points ${(\theta_s, \lambda_s, \tau_s)}_{s=1}^N$, the Lipschitz continuity of the relevant functions ensures that the obtained solution generalizes to the entire continuous domain $[0, t_f]$ of the ROP \eqref{eq:ROP}.

\begin{figure*}[!ht]
    \centering
    \includegraphics[width=\textwidth]{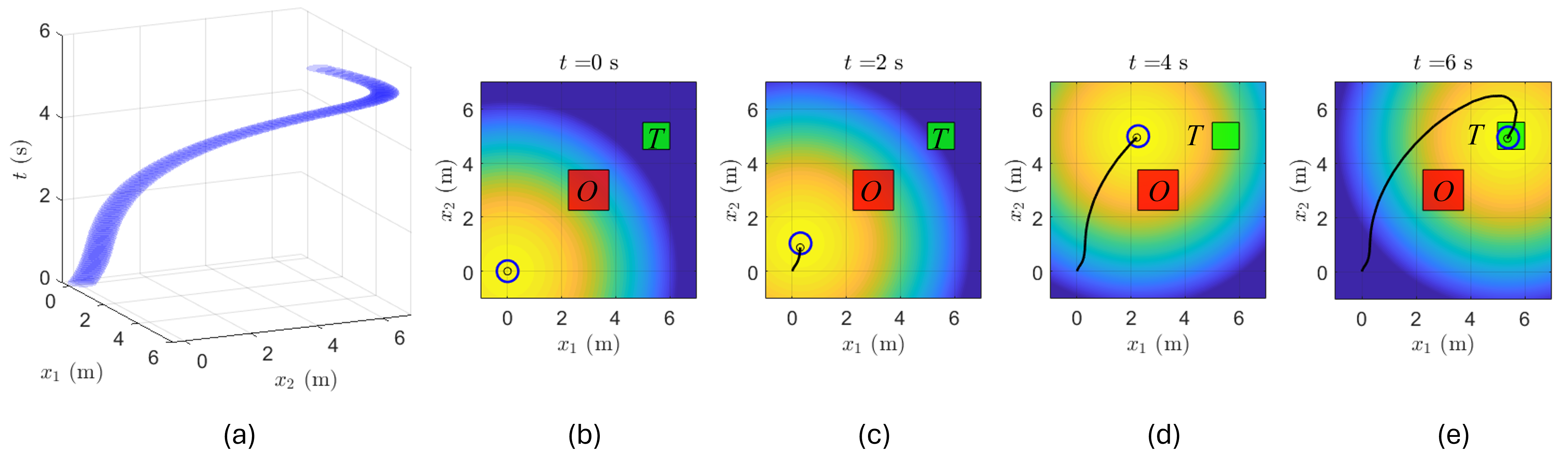}
    \caption{An omnidirectional mobile robot satisfying the STL task $\F_{[0,6]} T \land \G_{[0,6]} \neg O$, with the target $T$ (green) and obstacle $O$ (red). (a) The constructed STT (blue). (b)-(d) The 0-superlevel set of the TV-CBF (blue circle) and the controlled system trajectory (black line) at different time instances.}
    \label{fig:bf}
\end{figure*}

\begin{theorem} \label{th:constr}
    Consider the SOP in (\ref{eq:SOP}), solved using $N$ sampled points defined in \eqref{eq:ball}. Let $[q_\cen^*, q_\rad^*, \eta^*]$ be its optimal solution of SOP. If the following holds with $\mathcal{L} = \max\{\mathcal{L}_{\rad}, \sqrt{\mathcal{L}_\mu^2 + \mathcal{L}_\rho^2(\mathcal{L}_\cen+\mathcal{L}_\rad)^2}\}$: 
    \begin{equation} \label{eq:satisfy}
        \eta^* + \mathcal{L}\epsilon \leq 0,
    \end{equation}
    then the STT obtained from the SOP solution
    $$\Gamma(q_\cen^*, q_\rad^*, t) = \B(\cen(q_\cen^*, t), \rad(q_\rad^*, t)),$$
    satisfies all the conditions of Definition \ref{def:stt}. 
\end{theorem}

\begin{proof}\label{proof:constr}
This proof shows that under condition \eqref{eq:satisfy}, the STT, $\Gamma(q_\cen^*, q_\rad^*, t) = \B(\cen(q_\cen^*, t), \rad(q_\rad^*, t))$, obtained by solving the SOP in \eqref{eq:SOP}, satisfies all the conditions in Definition~\ref{def:stt}.

From \eqref{eq:ball}, for every $\tau \in [0, t_f]$, there exists a sampled point $\tau_s$ such that $|\tau - \tau_s| \leq \epsilon$. Therefore, for all $\tau \in [0, t_f]$:
\begin{align*}
    &-\rad(q_\rad^*, \tau) + \rad_d = -\rad(q_\rad^*, \tau_s) + \rad_d + \rad(q_\rad^*, \tau_s) - \rad(q_\rad^*, \tau) \\
    &\leq \ \eta^* + \mathcal{L}_\rad |\tau-\tau_s| \leq \eta^* + \mathcal{L}_\rad\varepsilon \leq \eta^* + \mathcal{L}\epsilon \leq 0.
\end{align*}
This implies that $\rad(q_\rad^*, \tau) \geq \rad_d > 0$ for all $\tau \in [0, t_f]$.

Similarly, from \eqref{eq:ball}, for every $(\theta, \lambda, \tau) \in [0,2\pi]^n \times [0,1] \times [0,t_f]$, there exists $(\theta_s, \lambda_s, \tau_s)$ such that $\|(\theta, \lambda, \tau) - (\theta_s, \lambda_s, \tau_s)\| \leq \epsilon$. Define the trajectories, $x, x_s: [0,t_f] \rightarrow \R^n$ as $x(\tau) = \cen(q_\cen^*, \tau) + \lambda \rad(q_\rad^*, \tau)s(\theta)$, $x_s(\tau_s) = \cen(q_\cen^*, \tau_s) + \lambda_s \rad(q_\rad^*, \tau_s)s(\theta_s)$. Let $\tilde{x}_s: [0,t_f] \rightarrow \R^n$ be the continuous reconstruction of $x_s$, such that at any time $\tau$ it holds the value from the nearest sample within $\epsilon$. Then,
\begin{align*}
     -&\rho^{\phi} (x ) = -\rho^{\phi} (x ) + \rho^{\phi} (\tilde x_s ) - \rho^{\phi} (\tilde x_s ) + \rho^{\phi} (x_s ) -\rho^{\phi} (x_s ) \\
     &= \| \mu(\theta, \lambda) - \mu(\theta_s, \lambda_s) \| \\ 
     &+ \max_{\tau \in [0,t_f], |\tau-\tau_s|\leq \epsilon} \mathcal{L}_\rho \|\hat x_s(\tau) - x_s(\tau_s)\| + \eta^* \\
     &\leq \ \eta^* + \mathcal{L}_\mu \|(\theta,\lambda)-(\theta_s,\lambda_s)\| + \mathcal{L}_\rho(\mathcal{L}_\cen+\mathcal{L}_\rad)|\tau-\tau_s| \\
     &\leq \eta^* + \sqrt{\mathcal{L}_\mu^2+\mathcal{L}_\rho^2 (\mathcal{L}_\cen+\mathcal{L}_\rad)^2}\epsilon \leq \eta^* + \mathcal{L}\epsilon \leq 0.
\end{align*}
Thus, $\rho^{\phi}(x) > 0$ for all $x : [0,t_f] \rightarrow \R^n$, such that, $x(\tau) \in \B(\cen(q_\cen^*, \tau), \rad(q_\rad^*, \tau)), \forall \tau \in [0, t_f]$.

Hence, when condition \eqref{eq:satisfy} is met, the STT $\Gamma(q_\cen^*, q_\rad^*, t)$, obtained by solving the SOP with finitely many samples, satisfies the conditions in Definition~\ref{def:stt}, completing the proof.
\end{proof}

\begin{remark}
The Lipschitz constants $\mathcal{L}\rad$ and $\mathcal{L}\cen$ required to verify condition~\eqref{eq:satisfy} can be estimated using the procedure described in \cite{DDSTT_arxiv}. The derivation of the Lipschitz constant for the STL robustness function, $\mathcal{L}_\rho$, is detailed in \cite{STL_Lip, STL_Lip_Sadegh}.
\end{remark}

\section{STT-CBF Framework}
This section presents the procedure for constructing TV-CBFs from the STT. 
We define the candidate TV-CBF as:
\begin{equation}
b(t,x) = 1 - \frac{1}{\rad^2(t)}\big(x-\cen(t)\big)^\top\big(x-\cen(t)\big),
\label{eq:barrier_function}
\end{equation}
where $\cen(t)$ and $\rad(t)$ are the STT center and radius at time $t$.

\begin{theorem}
The TV-CBF $b(t,x)$ defined in \eqref{eq:barrier_function} is continuously differentiable with respect to both $t$ and $x$.
\end{theorem}
\begin{proof}
From Definition~\ref{def:stt}, the STT center and radius, $\cen(t)$ and $\rad(t)$, are continuously differentiable functions of time $t$, with $\rad(t)>0$ for all $t \in [0,t_f]$. 
As $b(t,x)$ \eqref{eq:barrier_function} is obtained by composing and multiplying continuously differentiable functions of $t$ and $x$, it follows from standard results in multivariable calculus that $b(t,x)$ is continuously differentiable with respect to both $t$ and $x$.
\end{proof}

We now formulate a quadratic programming (QP) problem that ensures $\Cs(t)$ is forward invariant, 
\begin{align}\label{eq:QP}
&\min_{u}u^T K u \notag\\
\text{s.t. } &-\mathcal{L}_g b(t,x) u \leq \mathcal{L}_fb(t,x) + \frac{\partial b(t,x)}{\partial t} + \alpha(b(t,x)),
\end{align}
where $K \in \R^{n \times n}$ is a positive semi-definite matrix. 

\begin{theorem} 
Consider the STL formula $\phi$ in \eqref{eqn:stl} and the corresponding STT $\Gamma(t)$. Suppose that the TV-CBF $b(t,x)$ is defined as in Equation~\eqref{eq:barrier_function} with its 0-superlevel set: 
$$\Cs(t) = \{ x \in \R^n \mid b(t,x) \geq 0 \}.$$ 
Then $b(t,x)$ is a valid CBF over the time interval $[0,t_f]$, and under the control input $u^*(t,x)$, obtained by solving the QP \eqref{eq:QP}, the system satisfies the STL specification, i.e., $x \models \phi$. 
\end{theorem} 

\begin{figure}[t]
    \centering
    \includegraphics[width=0.48\textwidth]{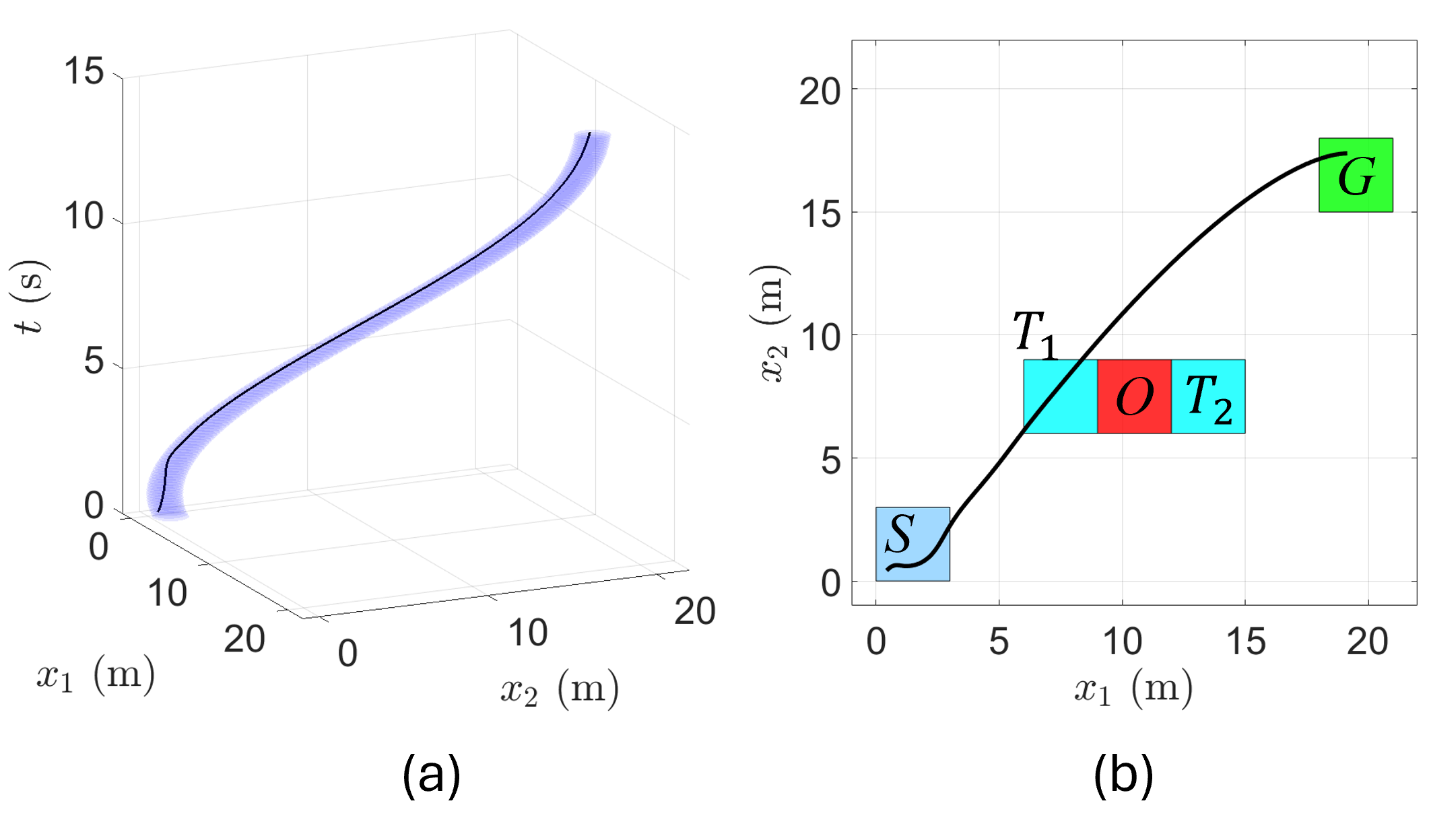}
    \caption{(a) STT (blue) and (b) the corresponding trajectory (black) for the differential-drive robot. The vehicle starts from $S$ (blue), visits $T_1$ or $T_2$ (cyan), then reaches $G$ (green), while avoiding obstacle $O$ (red) throughout the mission.}
    \label{fig:diff}
\end{figure}

\begin{figure*}[t]
    \centering
    \includegraphics[width=\textwidth]{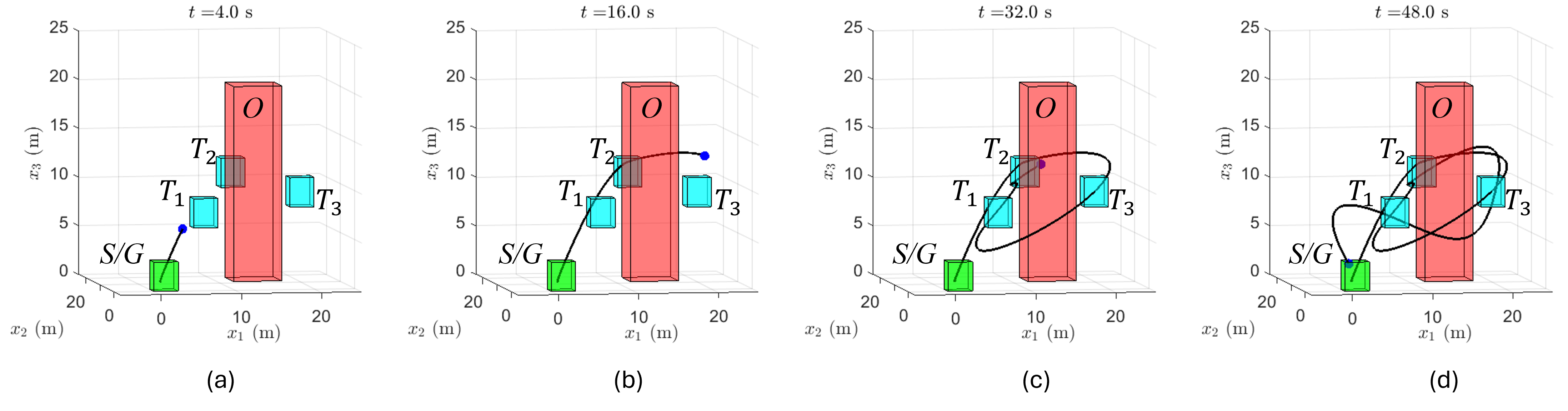}
    \caption{Quadrotor trajectory for the 3D surveillance task. The vehicle starts from $S$ (green), repeatedly visits $T_1$, $T_2$, and $T_3$ (cyan) in sequence, then reaches $G$ (green), while avoiding obstacle $O$ (red) throughout the mission. (a)–(d) show the vehicle trajectory at different time instances.}
    \label{fig:uav}
\end{figure*}


\begin{proof} 
From the definition of $b(t,x)$, we have: 
$$b(t,x) \geq 0 \iff \big(x-\cen(t)\big)^\top\big(x-\cen(t)\big) \leq \rad^2(t),$$
which implies that the 0-superlevel set of $b(t,x)$ is precisely the ball centered at $\cen(t)$ with radius $\rad(t)$, i.e., the STT $\Gamma(t)$. Thus, the relationship between $b(t,x)$ and $\Gamma(t)$ can be expressed as: 
\begin{subequations}
\begin{align} b(t,x) \geq 0, & \text{ if } x \in \Gamma(t) \\ b(t,x) = 0, & \text{ if } x \in \partial\Gamma(t) \\ b(t,x) > 0, & \text{ if } x \in \text{Int}\big(\Gamma(t)\big).
\end{align}
\end{subequations} 
Hence, the 0-superlevel set of $b(t,x)$ is equivalent to the tube:
$$\Cs(t) = \big \{x \in \R^n \mid \|x-\cen(t)\| \leq \rad(t) \big \} = \Gamma(t).$$ 
According to Definition~\ref{def:stt}, $\rad(t)$ is strictly positive for all $t \in [0, t_f]$. Consequently, $\Cs(t)$ is non-empty, ensuring that $b(t,x)$ defines a valid candidate TV-CBF.

Thus, from Theorem~\ref{thm:tvcbf}, the control input $u^*(t,x)$, obtained from the QP \eqref{eq:QP}, enforces forward invariance of $\Cs(t)$ under the system dynamics \eqref{eqn:sysdyn}:
$$x(0) \in \Cs(0) \implies x(t) \in \Cs(t) = \Gamma(t), \forall t \in [0,t_f].$$
Since the STT $\Gamma(t)$ is constructed such that $\rho^\phi(x) > 0$ for all $x \in \Gamma(t)$, maintaining the state within $\Gamma(t)$ ensures satisfaction of the STL formula.
Therefore, the closed-loop trajectory satisfies the STL specification, i.e., $x \models \phi$.
\end{proof}

To illustrate the TV-CBF construction from the STT, consider the 2D example in Figure~\ref{fig:bf}. The figure shows the temporal evolution of the STT $\Gamma(t)$ and the 0-superlevel set $\Cs(t)$ of the corresponding $b(t,x)$, representing the safe states where the barrier function is non-negative. In this example, an omnidirectional robot must reach the target region $T$ within $[0,6]$ s while avoiding the obstacle region $O$, satisfying the STL specification $\F_{[0,6]} T \land \G_{[0,6]} \neg O$.

\section{Case Studies}
We demonstrate the effectiveness of the proposed approach on two systems: a differential-drive mobile robot in a 2D environment and a quadrotor operating in a 3D environment. 
For each candidate value of $\eta$, the feasibility of the SOP is checked using the Z3 SMT solver \cite{Z3}. A near-minimal feasible value $\eta^*$ is then obtained via a bisection search over a user-defined interval until convergence. All case studies are implemented on a Linux Ubuntu system with an Intel i7-7700 CPU and 32 GB RAM. 
The implementation code is available at:
\href{https://github.com/RatnangshuD/STL_CBF_via_STT_main}{https://github.com/RatnangshuD/STL\_CBF\_via\_STT\_main}.

In case studies, each region $M$ is represented as a hyperrectangle with center $\mathrm{center}(M)$ and size $d_M$, described by the predicate $\tilde{M} := |x - \mathrm{center}(M)|_\infty < d_M$.

\subsection{Differential-Drive Mobile Robot}
We first consider a differential-drive mobile robot operating in a 2D workspace, as shown in Figure~\ref{fig:diff}. The system is assigned the following STL task: 
$$\psi = \tilde{S} \implies \F_{[6, 7]} (\tilde{T}_1 \vee \tilde{T}_2) \wedge \F_{[14,15]} \tilde{G} \wedge \G_{[0,15]} \neg \tilde{O}$$
with $S = [0,0.3] \times [0,0.3]$, $T_1 = [6,9] \times [6,9]$, $T_2 = [12,15] \times [6,9]$, $G = [18,21] \times [15,18]$, and $O = [9,12] \times [6,9]$.

This STL formula specifies a sequential temporal task where the robot, starting from the initial region $S$, must eventually reach either region $T_1$ or $T_2$ within 7–8 seconds, and subsequently proceed to the goal region $G$ within 14–15 seconds. Throughout the entire time horizon $[0,15]$ seconds, the robot must continuously avoid the obstacle region $O$. 
Figure~\ref{fig:diff} shows the generated STT and the system trajectories.

\subsection{Quadrotor in a 3D Environment}
We next consider a quadrotor \cite{Drone_example} performing a surveillance-like mission around a building, as shown in Figure~\ref{fig:uav}. The system is assigned the following STL task:
\begin{align*}
    \psi = &\G_{[0,42]} \big( (\tilde{S} \implies \F_{[6,7]} \tilde{T}_1) \wedge (\tilde{T}_1 \implies \F_{[6,7]} \tilde{T}_2) \\
    &\wedge (\tilde{T}_2 \implies \F_{[6,7]} \tilde{T}_3) \wedge (\tilde{T}_3 \implies \F_{[6,7]} \tilde{T}_1) \big) \\
    &\wedge \F_{[48,50]} \tilde{G} \wedge \G_{[0,50]} \neg \tilde{O},    
\end{align*}
with $S = G = [0,3]\times[0,3]\times[0,3]$, $T_1 = [6,9]\times[6,9]\times[0,3]$, $T_2 = [12,15]\times[21,24]\times[0,3]$, $T_3 = [18,21]\times[6,9]\times[0,3]$, and $O = [12,15]\times[12,15]\times[0,20]$.

This specification captures a persistent surveillance task, where the quadrotor, starting from the initial region $S$, must repeatedly visit the target regions $T_1$, $T_2$, and $T_3$ in sequence every 6–7 seconds until 42 seconds. After completing this, the quadrotor must reach the goal region $G$ between 48 and 50 seconds, while avoiding the obstacle region $O$ throughout the 50-second mission. Figure~\ref{fig:uav} shows the quadrotor trajectory.

\begin{figure*}[t]
    \centering
    \includegraphics[width=\textwidth]{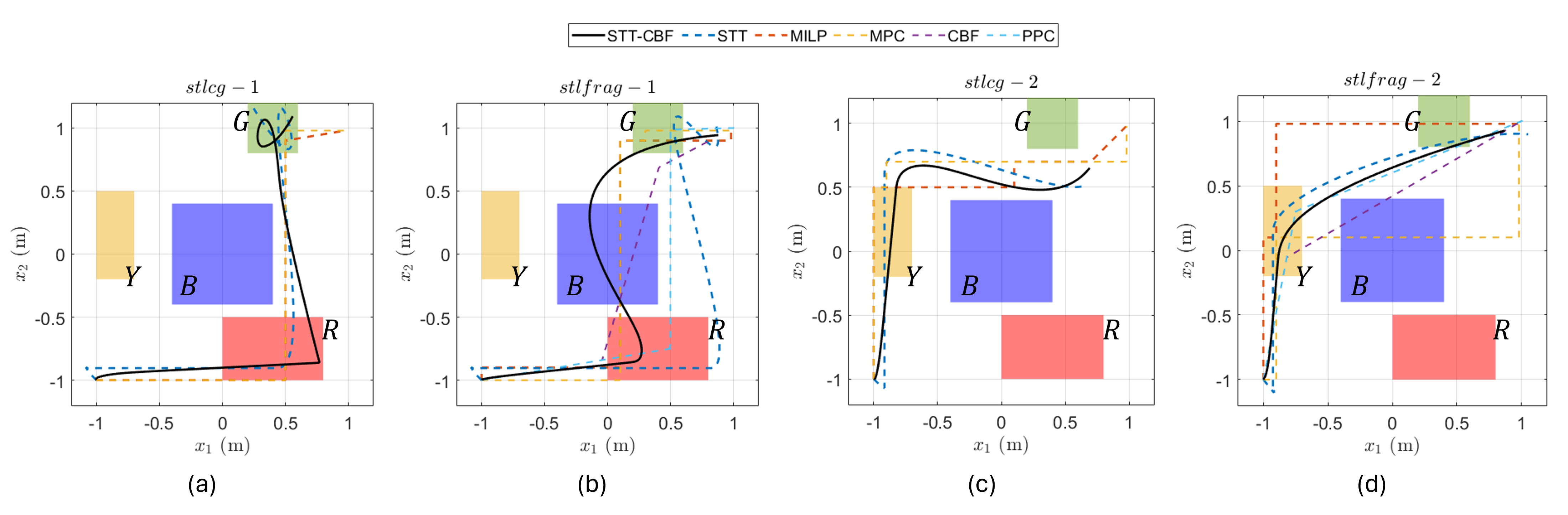}
    \caption{Comparison of trajectories obtained using the proposed STT-CBF framework, STT-based controller, MILP, MPC, CBF, and Funnel approaches for the four benchmark tasks (\textit{stlcg-1}, \textit{stlcg-2}, \textit{stlfrag-1}, and \textit{stlfrag-2}).}
    \label{fig:compare}
\end{figure*}

\begin{table}[!ht]
\centering
\caption{Computation time (seconds)}\label{table:computation_time}
\resizebox{0.48\textwidth}{!}{%
\begin{tabular}{|l|l|l|l|l|l|l|}
\hline
Case    & \textbf{STT-CBF} & STT & MILP & MPC  & CBF & PPC \\ \hline
\textit{stlcg-1} & \textbf{0.499} & 4.933  & 46.093  & 52.435  & N/A     & N/A        \\ \hline
\textit{stlfrag-1} & \textbf{0.416} & 0.205  & 6.905   & 7.458   & 3.774  & 0.142      \\ \hline
\textit{stlcg-2} & \textbf{0.426} & 8.743   & 84.814 & 101.280 & N/A     & N/A        \\ \hline
\textit{stlfrag-2} & \textbf{0.351} & 0.317   & 6.954  & 4.561  & 3.930  & 0.147      \\ \hline
\end{tabular}
}
\end{table}

\subsection{Comparison}
We evaluate our method on two benchmark STL tasks \cite{STL_multi_chuchu}, \textit{stlcg-1}: $\big(\F_{[0,10]}\G_{[0,5]}R \big) \wedge \big(\F_{[0,10]}\G_{[0,5]}G \big) \wedge \big(\G_{[0,10]}\neg B \big)$, \textit{stlcg-2}: $\big(\F_{[0,10]}\G_{[0,5]}Y \big) \wedge \big(\G_{[0,10]}\neg G \big) \wedge \big(\G_{[0,10]}\neg B \big)$. We also include their simplified STL fragment counterparts \cite{STL_PPC}, \textit{stlfrag-1}: $\big(\F_{[0,10]}R \big) \wedge \big(\F_{[0,10]}G \big)$, \textit{stlfrag-2}: $\F_{[0,10]}Y $.

Table~\ref{table:computation_time} reports the total computation time. The proposed spherical STT design not only enables smooth construction of TV-CBFs, but also significantly reduces synthesis time compared to the previous hyper-rectangular design \cite{das2025approximation}, due to fewer constraints and optimization variables. Compared to MILP \cite{STL_reactive_MILP} and MPC \cite{STL_MPC} approaches, STT-CBF achieves faster computation through efficient CBF-based synthesis. Unlike prior CBF \cite{CBF_STL} and PPC \cite{STL_funnel} methods, that handle only STL fragments, our framework addresses the full class of STL specifications. All computations were done with a time step of $0.05$~s, and the system trajectories are shown in Figure~\ref{fig:compare}.

\section{Conclusion}
This paper presented a unified framework for synthesizing TV-CBFs for general STL specifications using STTs. We proposed a spherical STT design and formulated its synthesis as a ROP using the robustness semantics of STL. By sampling over time and state, we reformulated it as a SOP and established conditions ensuring that its solution satisfies the full class of STL specifications. The resulting STTs were then used to construct TV-CBFs, guaranteeing STL satisfaction under any control law that renders them invariant. Case studies on a differential-drive robot and a quadrotor demonstrated the efficiency and scalability of the method, and comparison studies showed clear advantages over MILP-, MPC-, CBF- and PPC- based approaches. Future work will focus on extending the STT synthesis framework to explicitly incorporate input constraints and validate it further through hardware experiments.

\bibliographystyle{unsrt} 
\bibliography{sources} 

\appendix

\begin{proof}[of Theorem \ref{thm:lip_rob}]
    We show that $\mu(\theta, \lambda)
    := -\rho^{\phi} \big( x \big)$ with $x(\tau) = \cen(q_\cen, \tau) + \lambda \rad(q_\rad, \tau)s(\theta)$ for all $\tau \in [0,t_f]$, is Lipschitz continuous with respect to. \textit{(i)} $\theta$ and \textit{(ii)} $\lambda$,
    and then combine the results to obtain the Lipschitz constant.    
    
    From \cite{STL_Lip}, we know that there exists $\mathcal{L}_{\rho}\in \R^+$, such that
    $$| \rho^{\phi}(x_1, {.\kern 0em.\kern 0em.}, x_{k}, {.\kern 0em.\kern 0em.}, x_n) - \rho^{\phi}(x_1, {.\kern 0em.\kern 0em.}, \hat{x}_{k}, {.\kern 0em.\kern 0em.}, x_n) | \leq \mathcal{L}_{\rho} | {x}_{k} - \hat{x}_{k} |.$$

    Recall that $\|s(\theta)\|\leq 1$ for all $\theta$, and $\overline{\rad}:=\max_{t\in[0,t_f]}\rad(q_\rad,t)$ is finite because $\rad(\cdot)$ is Lipschitz on the compact interval $[0,t_f]$. Also let $\mathcal{L}_s$ be a Lipschitz bound for $s(\theta)$ (one can use $\mathcal{L}_s=\sqrt{n(n-1)}$ as a conservative bound).
     
    \begin{itemize}
        \item[\textit{(i)}] \textit{Lipschitz in $\theta$.} \\
        Fix some $(\lambda, \tau) \in [0,1] \times [0,t_f]$. For any $\theta,\hat\theta \in [0,2\pi]^n$,
        \begin{align*}
            & | \mu(\theta, \lambda) - \mu(\hat\theta, \lambda) | \\
            \leq  & \mathcal{L}_\rho \max_{\tau \in [0,t_f]} \| \lambda \rad(q_\rad, \tau) \big( s(\theta) - s(\hat\theta) \big) \|
            \leq \mathcal{L}_\rho \overline{\rad} \mathcal{L}_s \|\theta - \hat\theta\|.
        \end{align*}
        
        \item[\textit{(ii)}] \textit{Lipschitz in $\lambda$.} \\ 
        Fix some $(\theta, \tau) \in [0,2\pi]^n \times [0, t_f]$. For any $\lambda,\hat\lambda \in [0,1]$,
        \begin{align*}
            & | \mu(\theta, \lambda) - \mu(\theta, \hat\lambda) | \\
            \leq & \mathcal{L}_\rho  \max_{\tau \in [0,t_f]} \| (\lambda-\hat\lambda) \rad(q_\rad, \tau) s(\theta) \|
            \leq \mathcal{L}_\rho \overline{\rad} |\lambda - \hat\lambda|.
        \end{align*}

    \end{itemize}

    Now combining the above, for arbitrary $(\theta,\lambda), (\hat\theta,\hat\lambda) \in [0,2\pi]^n \times [0, 1]$, we can write
    \begin{align*}
        &| \mu(\theta, \lambda) - \mu(\hat\theta, \hat\lambda) | \\
        = & | \mu(\theta, \lambda) - \mu(\hat\theta, \lambda) 
        + \mu(\hat\theta, \lambda) - \mu(\hat\theta, \hat\lambda) \\ 
        &+ \mu(\hat\theta, \hat\lambda) - \mu(\hat\theta, \hat\lambda) | \\
        \leq & \mathcal{L}_\rho\mathcal{L}_s\overline{\rad}\|\theta - \hat\theta\| + \mathcal{L}_\rho\overline{\rad}|\lambda-\hat\lambda|.
    \end{align*}
    From Cauchy-Schwarz inequality,
    \begin{align}
        | \mu(\theta, \lambda) - \mu(\hat\theta, \hat\lambda) |\leq & \mathcal{L}_\mu \|(\theta, \lambda)-(\hat\theta, \hat\lambda)\|,
    \end{align}
    where 
    $\mathcal{L}_\mu = \mathcal{L}_\rho \overline{\rad}\sqrt{ (\mathcal{L}_s^2 + 1)}$.
\end{proof}

\end{document}